\documentclass[10pt, conference]{IEEEtran}
\IEEEoverridecommandlockouts

\usepackage[
backend=bibtex,
style=ieee,
backref=false,
]{biblatex}
\usepackage{graphicx}				

\usepackage[T1]{fontenc}

\usepackage{enumitem}

\DeclareFontEncoding{FMS}{}{}
\DeclareFontSubstitution{FMS}{futm}{m}{n}
\DeclareFontEncoding{FMX}{}{}
\DeclareFontSubstitution{FMX}{futm}{m}{n}
\DeclareSymbolFont{fouriersymbols}{FMS}{futm}{m}{n}
\DeclareSymbolFont{fourierlargesymbols}{FMX}{futm}{m}{n}
\DeclareMathDelimiter{\VERT}{\mathord}{fouriersymbols}{152}{fourierlargesymbols}{147}

\usepackage[english]{babel}
\usepackage{amsthm}
\usepackage{bbm}	
\usepackage{url}

\usepackage[
 colorlinks=true, bookmarks=true]{hyperref} 
\usepackage{ifthen}
\usepackage{tikz}
\usetikzlibrary{decorations.pathreplacing}
\usepackage{tikz-cd}

\usepackage{blindtext}
\usepackage{scrextend}
\addtokomafont{labelinglabel}{\sffamily}

\usepackage{amssymb}
\usepackage{amsmath}
\usepackage{cleveref}
\usepackage{mathrsfs}
\newcommand{\R}{\mathbb{R}}

\newcommand{\E}{\mathbb{E}}

\newcommand{\1}{\mathbf{1}}

\newcommand{\N}{\mathbb{N}}

\newcommand{\supp}{\operatorname{supp}}

\newcommand{\bern}{\operatorname{Bern}}

\newcommand{\norm}[1]{\left\lVert#1\right\rVert}

\newtheorem{theorem}{Theorem}[section]
\newtheorem{corollary}[theorem]{Corollary}
\newtheorem{lemma}[theorem]{Lemma}
\newtheorem{proposition}[theorem]{Proposition}
\theoremstyle{definition}
\newtheorem{definition}[theorem]{Definition}
\newtheorem{example}[theorem]{Example}

\newtheorem{remark}[theorem]{Remark}

\title{Beyond Hypercontractivity: Stronger Communication Converses via R\'enyi Contraction}

\author{\IEEEauthorblockN{Ryo Hommura Moore}
\IEEEauthorblockA{\textit{Okinawa Institute of Science and Technology} \\
Onna-son, Okinawa, Japan \\
ryo.hommura@oist.jp}
\and
\IEEEauthorblockN{Amedeo Roberto Esposito}
\IEEEauthorblockA{\textit{Okinawa Institute of Science and Technology} \\
Onna-son, Okinawa, Japan \\
amedeo.esposito@oist.jp}
}

\begin{document}
\renewcommand{\figurename}{Fig.}

\maketitle

\begin{abstract}
    We study the relationship between strong data-processing inequalities for R\'enyi divergences and hypercontractivity. For R\'enyi divergences of order greater than one, we characterize the contraction of R\'enyi divergences through a property of the dual Markov operator that we name power contraction. Using the contraction coefficient, we derive an upper bound on rectangle probabilities in terms of their marginals that is at least as strong as the corresponding bound obtained via hypercontractivity. Combining this result with inversion of the binary Rényi divergence, we derive communication lower bounds for one-way communication-assisted agreement distillation. Examples involving binary symmetric and rare-event Z-channels illustrate the improvements over standard hypercontractive converses.
\end{abstract}


\section{Introduction}
One of the fundamental results in information theory states that one cannot increase the amount of information in a signal or system by processing it; this is called the data-processing inequality (DPI). A sharpening of this result, known as the strong data-processing inequality (SDPI), may be used to quantify the amount of information contraction guaranteed by the DPI (cf. \cite[Ch. 33]{PolyanskiyWu2025}). Ahlswede and G\'acs \cite{AG76} have shown in their remarkable work that the SDPI coefficient of the Kullback-Leibler (KL) divergence has many deep connections with other important quantities, such as the maximal correlation and so-called hypercontractivity constants of the corresponding Markov operator. Beyond the KL divergence, Raginsky's work \cite{Raginsky2016} thoroughly studies SDPIs for $\Phi$-divergences. On the other hand, less is known regarding the SDPIs of R\'enyi-$\alpha$ divergences, since the nature of the contraction provided by these inequalities seems to be quite different. Some properties of these SDPIs have been studied in \cite{JEG_2024, VEG_contraction}.
For R\'enyi divergences of order $\alpha>1$, strong data-processing inequalities admit a useful operator interpretation. Let $K$ be a Markov kernel with reference input distribution $\mu$, and let $K_\mu^\star$ denote the associated dual operator transporting input densities to output densities. The inequality $D_\alpha(\nu K\|\mu K)\le\theta D_\alpha(\nu\|\mu)$ is equivalent to
$
\|K^\star f\|_{L^\alpha(\mu K)} \le \|f\|_{L^\alpha(\mu)}^\theta
$
for every nonnegative $f$ with $\int f\,\mathrm{d}\mu=1$. We call this property $L^\alpha$-power contraction when $0<\theta<1$. It describes contraction through a power $\theta$ of the input $L^\alpha$-norm with $0<\theta<1$, whereas hypercontractivity controls the output $L^\alpha$-norm by an input $L^\beta$-norm with $\beta<\alpha$.

Building on the relationship between hypercontractivity and R\'enyi divergence in~\cite[§IV]{Raginsky2013}, we compare these two operator properties. Hypercontractivity implies power contraction with an explicit exponent, but the converse need not hold. We use the optimal R\'enyi SDPI coefficient to bound rectangle probabilities $P_{XY}(A\times B)$ in terms of their marginals. The resulting bound is no weaker than the corresponding one induced by hypercontractivity.

We apply this bound to communication-assisted agreement distillation~\cite{AG76,GR16,Zhao2022}, in which two parties observe correlated samples and seek to agree on a uniformly distributed binary string using limited communication. Combining R\'enyi SDPI coefficients with binary inversion, we obtain a refined communication converse. We then derive an explicit relaxation that expresses the communication lower bound directly in terms of the SDPI coefficient. Even after this relaxation, the bound is no weaker than the corresponding hypercontractive converse and can be strictly sharper. We conclude with numerical examples illustrating these improvements.

\section{Preliminaries}


We will use the following notations and conventions in this paper: First, given a nonzero real number $p$, we denote $p'$ to be the H\"older conjugate of $p$, i.e. if $p \neq 1$, $p' = \frac{p}{p-1}$ so that $\frac{1}{p} + \frac{1}{p'} = 1$, and if $p = 1$, we denote $p' = \infty$. Second, the logarithmic function will be of base 2 (i.e. $\log(x) := \log_2(x)$).

\subsection{Markov kernels}
Let $(\Omega, \mathcal{F}, \mu)$ be a probability space. We recall that $K: \mathcal{F} \times \Omega \to [0, 1]$ is called a \textit{Markov kernel} (or \textit{channel}) if
for every $x \in \Omega$, the map $E \in \mathcal{F} \mapsto K(E|x)$ is a probability measure on $(\Omega, \mathcal{F})$, and for every $E \in \mathcal{F}$, the mapping $x \in \Omega \mapsto K(E|x)$ is an $\mathcal{F}$-measurable function.

For every $E \in \mathcal{F}$, we denote $\mu K(E) = \int K(E|x) \, \mathrm{d}\mu(x)$, and for any $p \in (0, \infty)$ and for every bounded and measurable function $f$, we denote $Kf(x) := \int f(y) \, \mathrm{d}K(y|x)$. Furthermore, if $p > 1$, Jensen's inequality tells us that
$
\norm{Kf}_{L^p(\mu)}^p 
\leq \norm{f}_{L^p(\mu K)}^p \, , 
$
which implies that the map $K$ can be extended to a map from $L^p(\mu K)$ to $L^p(\mu)$. 

Suppose that $\mu \otimes K $ is a joint distribution of a random pair $(X, Y) \in \mathsf{X} \times \mathsf{Y}$ such that $P_X = \mu$ and $P_{Y|X} = K$ (hence $P_Y = \mu K$). Then $Kf(x) = \E[f(Y)|X = x]$ for any $f \in L^1(\mu K)$ and for $\mu$-a.e. $x \in \mathsf{X}$. 

Given a Markov kernel $K$ and a probability measure $\mu$, we can define the dual (adjoint) operator $K_\mu^\star: L^{p'}(\mu) \to L^{p'}(\mu K)$. In particular, if $f \in L^p(\mu K)$ and $g \in L^{p'}(\mu)$, then
$\int Kf \cdot g \, \mathrm{d}\mu = \int f \cdot K_\mu^\star g \, \mathrm{d}\mu K \, . $
When $(\Omega, \mathcal{F}, \mu)$ is a finite probability space in which $\mu K$ has full support, we have (see \cite[(1.1)]{Raginsky2016})
\[ K^\star_\mu(x|y) = \frac{K(y|x)\mu(x)}{\mu K(y)} \, . \]

\subsection{R\'enyi Divergence, SDPI Constant}\label{ss:Renyi}
Let $\alpha \in (0, 1) \cup (1, \infty)$ be given. We recall that the \textit{R\'enyi divergence of order $\alpha$} between two probability measures $\mu$ and $\nu$ on $\mathsf{X}$ for which $\nu \ll \mu$ is defined so that we have
\[ D_\alpha(\nu || \mu) := \frac{1}{\alpha - 1} \log \E_\mu \left[ \left(\frac{\mathrm{d}\nu}{\mathrm{d}\mu}\right)^\alpha \right] = \frac{\alpha}{\alpha-1} \log \norm{\frac{\mathrm{d}\nu}{\mathrm{d}\mu}}_\alpha \, .\]

Elementary properties of R\'enyi divergence are summarized in \cite{VH_Renyi}. To mention a few notable facts, we know that $D_\alpha$ is non-decreasing in $\alpha$, continuous in $\alpha$ on $\R_{\geq 0}$ whenever $D_\alpha < \infty$ (so $D_\infty$ can be defined by evaluating the limit), and the KL divergence $D(\nu||\mu)$ corresponds to the limit $\alpha \to 1$ provided that $D_\beta < \infty$ for some $\beta > 1$.

We also have the \textit{data-processing inequality (DPI)} for R\'enyi-$\alpha$ divergences, which says that given a Markov kernel $K$, for any $\alpha \in [0, \infty]$ we have
    \begin{equation}\label{DPI}
        D_\alpha(\nu K || \mu K) \leq D_\alpha (\nu || \mu) \, .
    \end{equation}
Let us discuss (\ref{DPI}) a bit further; we fix a probability measure $\mu$, and suppose that there exists $c \in [0, 1]$ such that for every probability measure $\nu \ll \mu$, we have
$D_\alpha(\nu K || \mu K) \leq c D_\alpha(\nu || \mu)$.
In this case, we say $K$ satisfies \textit{the R\'enyi-$\alpha$ strong data-processing inequality (SDPI) at $\mu$ with constant $c$}, and we define the \textit{SDPI constant} $\eta_\alpha(\mu, K)$ for R\'enyi $\alpha$-divergence to be the infimum of $c \in [0, 1]$ for which the estimate is satisfied; alternatively, and more formally,
\[\eta_\alpha(\mu, K) := \sup_{\nu: 0 < D_\alpha(\nu||\mu) < \infty} \frac{D_\alpha(\nu K|| \mu K)}{D_\alpha(\nu||\mu)} \, . \]

\section{$L^\alpha$-Power contraction}
We discuss the notion of contraction that one can obtain from R\'enyi divergences. In fact, we will prove a result in a more general functional analytic setting, followed by an example with Markov kernels. First, we define such contraction:
\begin{definition}
    Let $(\mathsf{X}, \mu)$ and $(\mathsf{Y}, \rho)$ be two measure spaces, where we assume $\mu$ to be a probability measure. Let $\alpha \in (1, \infty)$. Suppose $T: L^\alpha(\mu) \to L^\alpha(\rho)$ is a positive linear operator such that there exists $c \in (0, 1)$ such that for every mean-1, non-negative $f \in L^\alpha(\mu)$, we have $\norm{Tf}_\alpha \leq \norm{f}_\alpha^c \, .$ Then we say $T$ is an \textit{$L^\alpha$-power contraction (with exponent $c$)}.
\end{definition}
The next theorem will show how an $L^\alpha$-power contraction contracts any function in $L^\alpha(\mu)$. 
\begin{theorem}\label{t:expContraction}
Let $(\mathsf{X}, \mu)$ and $(\mathsf{Y}, \rho)$ be two measure spaces, where we assume $\mu$ to be a probability measure. Let $\alpha \in (1, \infty)$. Suppose that $T: L^\alpha(\mu) \to L^\alpha(\rho)$ is an $L^\alpha$-power contraction with exponent $c \in (0, 1)$. Then for every $g \in L^\alpha(\mu)$, we have
\begin{equation}
    \label{eq:expContraction}
    \norm{Tg}_{L^\alpha(\rho)} \leq \norm{g}_{L^1(\mu)}^{1-c} \cdot \norm{g}_{L^\alpha(\mu)}^c \, .
\end{equation}
Moreover, $T$ is a contraction (hence bounded).
\end{theorem}
\begin{proof}
    The inequality is trivial if $g \equiv 0$, so we assume otherwise.
    First, consider a non-negative function $f \in L^\alpha(\mu)$ for which $\norm{f}_1 \neq 0$. Let $h = f / \norm{f}_1$. Then $\norm{h}_1 = 1$, so the hypothesis of the theorem implies that $\norm{Th}_\alpha \leq \norm{h}_\alpha^c$, or equivalently,
    \[ \frac{\norm{Tf}_\alpha}{\norm{f}_1}  \leq \left(\frac{\norm{f}_\alpha}{\norm{f}_1}\right)^c \, . \]
    Multiply both sides by $\norm{f}_1$, and we get (\ref{eq:expContraction}) for non-negative functions.

    Let $g \in L^\alpha(\mu)$ be any function. Since $T$ is positive, we have $|Tg| \leq T|g|$. Since $|g|$ is non-negative, we have
    $\norm{Tg}_\alpha \leq \norm{T|g|}_\alpha \leq \norm{g}_1^{1-c} \cdot \norm{g}_\alpha^c .$

    Apply Jensen's inequality, $\norm{g}_1 \leq \norm{g}_\alpha \, $, to (\ref{eq:expContraction}) to observe that $T$ is a contraction.
\end{proof}
An important example of an $L^\alpha$-power contraction is the dual of a Markov kernel, and we can identify the optimal exponent.
\begin{example}\label{ex:SDPI}
Choose any probability measure $\nu \ll \mu$ for which $f := \frac{\mathrm{d}\nu}{\mathrm{d}\mu} \in L^\alpha(\mu)$. It is well known that $K^\star f = \frac{\mathrm{d}\nu K}{\mathrm{d}\mu K}$, $\mu  K$-a.e. (cf. \cite[Lemma 2.2]{EM_contraction}), so we have
$\norm{K_\mu^\star f}_{L^\alpha(\mu K)} \leq \norm{f}_{L^\alpha(\mu)}^{\eta_\alpha(\mu, K)} \, .$
Therefore, $K^\star_\mu$ is an $L^\alpha$-power contraction with exponent $\eta_\alpha(\mu, K)$ (assuming that this quantity is in $(0, 1)$). In fact, an alternative characterization of $\eta_\alpha(\mu, K)$ is
$\eta_\alpha(\mu, K)= \inf \bigl\{ c \in [0, 1]: \norm{K^\star_\mu f}_{L^\alpha(\mu K)} \leq \norm{f}_{L^1(\mu)}^{1-c}\norm{f}_{L^\alpha(\mu)}^c \text{for all } f \in L^\alpha(\mu) \bigr\} \, .
$
\end{example}
The following theorem provides an alternative characterization of a Markov kernel whose dual is an $L^\alpha$-power contraction.
\begin{theorem}
    \label{t:HolderRenyi}
    Let $K$ be a Markov kernel on a probability space $(\Omega, \mathcal{F}, \mu)$, and let $\alpha > 1$, and let $c \in (0, 1)$ be fixed. The following statements are equivalent.
    \begin{enumerate}[label = (\alph*)]
        \item The dual operator $K_\mu^\star: L^\alpha(\mu) \to L^\alpha(\mu K)$ is an $L^\alpha$-power contraction with exponent $c$.
        \item For every $f \in L^\alpha(\mu)$ and $g \in L^{\alpha'}(\mu K)$, we have
        \begin{align*} &\E_{\mu \otimes K} [f(X)g(Y)] \\
        &\leq \norm{f(X)}_{L^1(\mu)}^{1-c} \norm{f(X)}_{L^\alpha(\mu)}^c \norm{g(Y)}_{L^{\alpha'}(\mu K)} \, . \end{align*}
    \end{enumerate}
\end{theorem}
\begin{proof}
    First we prove that (a) implies (b). By the tower property of conditional expectations and H\"older's inequality, we obtain
    \begin{align*}
        \E_{\mu \otimes K}[f(X)g(Y)] 
        &= \E_{\mu K}[\E_{\mu \otimes K}[f(X)g(Y)|Y]] \\
        &= \E_{\mu K}[K^\star_\mu f(Y) g(Y)] \\
        &\leq \norm{K^\star_\mu f(Y)}_{L^\alpha(\mu K)} \norm{ g(Y)}_{L^{\alpha'}(\mu K)} \, .
    \end{align*}
    We apply Theorem \ref{t:expContraction} to show that (b) holds.

    Next, we prove that (b) implies (a). Let $f \in L^\alpha(\mu)$ be a non-negative function for which $\E_\mu f = 1$ (which implies that $\norm{f}_{L^1(\mu)} = 1)$. By the tower property, we have
    \begin{align*}
        \norm{K^\star_\mu f(Y)}_{L^\alpha(\mu K)}^\alpha
        &= \E_{\mu K}[K_\mu^\star f(Y) \cdot K_\mu^\star f(Y)^{\alpha-1}] \\
        &= \E_{\mu  K}[\E_{\mu \otimes K}[f(X)|Y] \cdot K_\mu^\star f(Y)^{\alpha - 1}] \\
        &= \E_{\mu \otimes  K}[f(X)K_\mu^\star f(Y)^{\alpha-1}] \, .
    \end{align*}
    We may apply our hypothesis provided that $(K_\mu^\star f)^{\alpha - 1} \in L^{\alpha'}(\mu K)$. To verify this, recall that the Markov kernel $K$ is a contraction, and so is $K_\mu^\star$. Thus, we apply Jensen's inequality to observe that
    $\E_{\mu K}[K_\mu^\star f(Y)^{(\alpha-1)\cdot \frac{\alpha}{\alpha - 1}}]= \E_{\mu K}[K_\mu^\star f(Y)^\alpha] \leq \E_{\mu}[f(X)^\alpha] < \infty \, .  $
    Hence, we apply our hypothesis to see that
    \begin{align*} 
    \norm{K_\mu^\star f(Y)}_{L^\alpha(\mu  K)}^\alpha 
    &\leq \norm{f(X)}_{L^\alpha(\mu)}^{c} \norm{K_\mu^\star f(Y)}_{L^\alpha(\mu K)}^{\alpha/\alpha'} \, .  \end{align*}
    By dividing both sides by $\norm{K_\mu^\star f(Y)}_{L^\alpha(\mu K)}^{\alpha/\alpha'}$ (which is positive), and observing that $1 - \frac{1}{\alpha'} = \frac{1}{\alpha}\,$, we obtain (a).
\end{proof}
One can use this result to obtain an upper bound for the rectangle probability.
\begin{corollary}\label{rectangleBound}
    For any measurable sets $A \subset \mathsf{X}$ and $B \subset \mathsf{Y}$, we have
    \begin{equation}
        \mu \otimes K(A \times B) \leq \mu(A)^{1 - \eta_\alpha(\mu, K)/\alpha'} \mu K(B)^{1/\alpha'} \, .
    \end{equation}
\end{corollary}

\begin{proof}
    Set $f = \1_A$ and $g = \1_B$ in Theorem \ref{t:HolderRenyi}, and choose $c = \eta_\alpha(\mu, K)$ as we discussed in Example \ref{ex:SDPI}.
\end{proof}
Corollary~\ref{rectangleBound} improves over the corresponding result given by hypercontractivity, as we will see in \S\ref{sec:RDH}.
\section{R\'enyi divergence and hypercontractivity}\label{sec:RDH}
Let $\alpha, \beta \in \R$ such that $1 \leq \beta < \alpha$, and suppose that $T: L^\alpha \to L^\alpha$ is a linear operator. Recall that $T$ is \textit{$(\alpha, \beta)$-hypercontractive} if for any $f \in L^\alpha(\mu)$, we have
$
    \norm{Tf}_\alpha \leq \norm{f}_\beta \, .
$

Here is a necessary condition for hypercontractivity, that relates to the type of operator discussed in Theorem \ref{t:expContraction}.
\begin{proposition}\label{T:Nec}
    Let $(\mathsf{X}, \mu)$ be a probability space, and let $(\mathsf{Y}, \nu)$ be a measure space. Suppose that $T: L^\alpha(\mu) \to L^\alpha(\nu)$ is a positive linear operator. If $T$ is $(\alpha, \beta)$-hypercontractive for some $\beta \in (1, \alpha)$,
    then $T$ is an $L^\alpha$-power contraction with exponent
    $c = \frac{\alpha'}{\beta'} =  \frac{\alpha(\beta-1)}{\beta(\alpha-1)} \, .$
\end{proposition}

\begin{proof}
    By the $(\alpha, \beta)$-hypercontractivity and the Riesz-Thorin interpolation, we must have $\norm{Tf}_\alpha \leq \norm{f}_\beta \leq \norm{f}_1^{1-c}\norm{f}_\alpha^c \, ,$ where $c = \frac{\alpha(\beta-1)}{\beta(\alpha-1)} \,. $
\end{proof}
We will discuss the relationships between R\'enyi-SDPI constants and hypercontractivity in the setting of Markov kernels. In the remainder of the paper, given a Markov kernel $K$, we denote 
$$\begin{aligned}
\beta_{K}(\alpha) &:= \inf \bigl\{\beta \in [1, \alpha] : \norm{K^\star f}_{L^\alpha(\mu K)} \leq \norm{f}_{L^\beta(\mu)} 
\\ &\text{ for every } f \in L^\alpha(\mu) \bigr\}. \end{aligned}$$

First, Raginsky \cite[Proof of Theorem 6]{Raginsky2013} showed that if $K^\star_\mu$ is $(\alpha, \beta)$-hypercontractive, then for every $\nu \ll \mu$ for which $\nu \neq \mu$, we have 
        \[ \frac{D_\alpha(\nu K || \mu K)}{D_\beta(\nu||\mu)} \leq \frac{\alpha(\beta-1)}{\beta(\alpha-1)} \, . \]
    Recall that $p \mapsto D_p(\nu||\mu)$ is a non-decreasing function. Furthermore, for fixed $\alpha > 1$, the function $\beta \mapsto \frac{\alpha(\beta-1)}{\beta(\alpha-1)}$ is increasing in $\beta$ on $[1, \alpha)$. Combining these facts, we have
        \begin{equation}
            \label{SDPI_bound}
            \eta_\alpha(\mu, K) \leq h_\alpha(\mu, K) := \frac{\alpha'}{\beta_K(\alpha)'} \, .
        \end{equation}
    However, the fixed-parameter condition \(\eta_\alpha(\mu,K)\leq\alpha'/\beta'\) is not sufficient for \((\alpha,\beta)\)-hypercontractivity. For example, one can observe this if $K_\lambda$ is the doubly symmetric binary symmetric channel (DSBS) with flip rate $\lambda$. Indeed, recall that by the Bonami-Beckner inequality (cf. \cite[``The Hypercontractivity Theorem", p. 284]{O’Donnell_2014}), combined with \cite[Theorem 3.1]{KA_2016}, we have 
    \begin{equation}\label{DSBS}
    \beta_{K_\lambda}(\alpha)   =  1 + (1-2\lambda)^2(\alpha - 1) \, . 
    \end{equation}
    Notice that $K_{0.1}$ fails to be $(2, 1.6)$-hypercontractive since $\beta_{K_{0.1}}(2) = 1.64$. Nevertheless, by \cite[Corollary 6]{JEG_2024}, we have $\eta_2(\bern(1/2), K_{0.1}) \approx 0.7137 < 2'/1.6' = 0.75$.

Second, if $X$ is a random variable, and $Y$ is the outcome of $X$ through the Markov kernel $K$, the $(\alpha, \beta)$-hypercontractivity of $K^\star_\mu$ implies that for every $f \in L^{\beta}(\mu)$ and for any $g \in L^{\alpha'}(\mu K)$, we have
\begin{equation}
    \label{eq:HolderHC}
    \E_{P_{XY}}[f(X)g(Y)] \leq \norm{f(X)}_{L^{\beta}(P_X)} \norm{g(Y)}_{L^{\alpha'}(P_Y)} \, ,
\end{equation}
so in particular, if we set $f = \1_A$ and $g = \1_B$ for measurable sets $A \subset \mathsf{X}$ and $B \subset \mathsf{Y}$, we have
\begin{equation}\label{HC-bound1}
     P_{XY}(X \in A, Y \in B) \leq P_X(A)^{1/\beta} P_Y(B)^{1/\alpha'} \, . 
\end{equation}
Recall that Corollary \ref{rectangleBound} says that we have
\begin{equation}\label{ineq:JD-bound}
P_{XY}(X \in A, Y \in B) \leq P_X(A)^{1 - \frac{\eta_\alpha(\mu, K)}{\alpha'}} P_Y(B)^{\frac{1}{\alpha'}} \, . 
\end{equation}
By (\ref{SDPI_bound}), for any $\beta \in [\beta_K(\alpha), \alpha)$, we have
    \begin{equation}\label{Renyi_HC_comparison}
    1 - \frac{\eta_\alpha(\mu, K)}{\alpha'} \geq \frac{1}{\beta} \, ,\end{equation}
which implies that the upper bound on the rectangle probability (\ref{ineq:JD-bound}) is never weaker than the one given in (\ref{HC-bound1}).

\begin{remark}
We recall that the hypercontractivity of Markov kernels tensorizes in the following sense: Given $\alpha > 1$ and a Markov kernel $K$, we define
$ s_\alpha(K) := \beta_K(\alpha)/\alpha$.
Now let $\alpha > 1$, and consider several Markov kernels $K_1, \ldots, K_n$ with input law $\mu$. By \cite[Theorem 2]{AG76}, we have
$s_\alpha\left(\bigotimes_{i=1}^n K_i\right) = \max_{1 \leq i \leq n} s_\alpha(K_i) \, ,$ so by applying (\ref{SDPI_bound}), we obtain
\begin{equation}\label{SDPI-HC-tensorization}
\eta_\alpha\left(\mu^{\otimes n}, \bigotimes_{i=1}^n K_i \right) \leq \max_{1 \leq i \leq n} h_\alpha(\mu, K_i)\, . \end{equation}
\end{remark}

\section{Case study: Communication assisted agreement distillation}
In this section, we compare hypercontractivity and R\'enyi SDPI in the \textit{communication assisted agreement distillation} problem. Guruswami and Radhakrishnan have shown that the hypercontractivity property can be used to establish impossibility results for binary symmetric channels (BSCs) and some binary erasure channels (BECs) \cite{GR16}; Zhao has extended their results to general channels \cite[Chapter 3]{Zhao2022}.

\subsection{Setting}\label{settings}
We describe the framework of the communication assisted agreement distillation that we will discuss; for the remainder of this section of the paper, we will use the settings and the notations introduced here. 
Let $n,k\in\mathbb N$. Alice observes a binary string $X^n\sim P_{X^n} = \mu^{\otimes n} $, and Bob observes $Y^n \sim P_{Y^n} = \mu^{\otimes n} K^{\otimes n} = (\mu K)^{\otimes n}$, obtained by passing $X^n$ through the memoryless binary channel $K$. Thus,
$
(X^n,Y^n)\sim P_{X^nY^n} = (\mu\otimes K)^{\otimes n}.
$

Alice and Bob wish to agree on a uniformly distributed binary string of length $k$, using one-way communication from Alice to Bob. Alice first applies a deterministic map
$
g_A:\{0,1\}^n\longrightarrow\{0,1\}^k
$
to produce $Z=g_A(X^n)$. We require that $Z$ be exactly uniform: for every $z\in\{0,1\}^k$,
$
\mu^{\otimes n}(A_z)=2^{-k},
\text{ where }
A_z:=\{x\in\{0,1\}^n:g_A(x)=z\}.
$

For a communication fraction $c\in(0,1)$, let $m=\lfloor ck\rfloor$. Alice sends Bob the message
$
M=f_c(X^n),
\text{ where }
f_c:\{0,1\}^n\longrightarrow\{0,1\}^{m}.
$
Bob then combines this message with his observation $Y^n$, using a decoder
$
g_{B,c}:\{0,1\}^n\times\{0,1\}^{m}
\longrightarrow\{0,1\}^k,
$
to produce the estimate
$
\widehat Z=g_{B,c}(Y^n,M).
$ Let $z \in \{0, 1\}^k$ and $B_z := \{y^n \in \{0, 1\}^n: \text{$\exists$ $u \in \{0, 1\}^m$ such that }g_{B, c}(y^n, u) = z\}$. For each $y \in \{0, 1\}^n$, the number of codes that Bob's decoder may produce is at most $2^{ck}$, which implies that 
\begin{equation}\label{Bz_sum_bound}
\sum_{z \in \{0, 1\}^k}  (\mu K)^{\otimes n}(B_z) \leq 2^{ck}.
\end{equation} The message $M$ need not consist of selected bits of $Z$; it may be any prescribed function of Alice’s observation.

The protocol succeeds when Bob recovers Alice’s entire output, that is, when $\widehat Z=Z$. The desired guarantee is that, for a given $\gamma>0$, there exist $c=c_\gamma\in(0,1)$, an encoder $f_c$, and a decoder $g_{B,c}$ such that
$$
P_{X^nY^n}\!\left[
g_A(X^n)=g_{B,c}\bigl(Y^n,f_c(X^n)\bigr)
\right]\ge 2^{-\gamma k}.
$$

\subsection{Communication converses}
We will present a way of computing the communication converse (i.e. a lower bound on $c$) using R\'enyi SDPI and binary R\'enyi inversion. 

In the remainder of the paper, the R\'enyi-$\alpha$ divergence of binary distributions plays a central role. For any $\alpha > 1$, we define $d_\alpha: [0, 1]^2 \to \R_{\geq 0} \cup \{+\infty\}$ such that
 \begin{align}
    d_\alpha(p||q) \label{d_alpha}
    &:= D_\alpha(\bern(p)||\bern(q)) \\
    &= \frac{1}{\alpha - 1} \log\left[ p^\alpha q^{1-\alpha} + (1-p)^\alpha(1-q)^{1-\alpha} \right] \, . \nonumber
\end{align}

We first prove a lemma.
\begin{lemma}\label{BIE_lemma}
In the setting of \S\ref{settings}, let
$
E:=\bigcup_{z\in\{0,1\}^k}A_z\times B_z$,
$p:=P_{X^nY^n}(E)$, and
$q:=(P_{X^n}\otimes P_{Y^n})(E)$.
Then
$
p\ge 2^{-\gamma k}$, $
q\le 2^{-k(1-c)},
$
and for every \(\alpha>1\),
\[
d_\alpha(p\Vert q)
\le k\eta_\alpha(\mu^{\otimes n},K^{\otimes n}).
\]
\end{lemma}

\begin{proof}
The successful-agreement event is contained in \(E\), which gives the lower bound on \(p\). By (\ref{Bz_sum_bound}), we have
\[
q=2^{-k}\sum_z P_{Y^n}(B_z)
\le 2^{-k(1-c)}.
\]
Put
$
p_z:=K^{\otimes n}(B_z\mid A_z)$ and $
q_z:=P_{Y^n}(B_z)$, and we denote $(\mu^{\otimes n})_{A_z}$ to be the conditional probability measure given $A_z$, i.e. $\mathrm{d}(\mu^{\otimes n})_{A_z} := \frac{\1_{A_z}}{\mu^{\otimes n}(A_z)} \mathrm{d} \mu^{\otimes n}$. Then uniformity gives
$
D_\alpha((\mu^{\otimes n})_{A_z}\Vert\mu^{\otimes n})=k.
$
Consequently, the data-processing inequality (\ref{DPI}) as well as the definition of SDPI yields
\begin{align*}
d_\alpha(p_z\Vert q_z)
&\le
D_\alpha\!\left(
(\mu^{\otimes n})_{A_z}K^{\otimes n}
\,\middle\Vert\,(\mu K)^{\otimes n}
\right)\\
&\le k\eta_\alpha(\mu^{\otimes n},K^{\otimes n}).
\end{align*}
Moreover, \(p=2^{-k}\sum_zp_z\) and \(q=2^{-k}\sum_zq_z\). By the joint convexity of
$
(u,v)\longmapsto
u^\alpha v^{1-\alpha}
+(1-u)^\alpha(1-v)^{1-\alpha},
$
we obtain
\begin{align*}
2^{(\alpha-1)d_\alpha(p\Vert q)}
&\le
2^{-k}\sum_z2^{(\alpha-1)d_\alpha(p_z\Vert q_z)}\\
&\le
2^{(\alpha-1)k\eta_\alpha(\mu^{\otimes n},K^{\otimes n})}.
\end{align*}
Taking logarithms proves the claim.
\end{proof}
Lemma \ref{BIE_lemma} connects the probability of successful agreement to the amount of communication: The success requirement gives a lower bound on \(P_{X^nY^n}(E)\), while the number of bits Alice sends gives an upper bound on \((P_{X^n}\otimes P_{Y^n})(E)\). The Rényi SDPI coefficient limits how far apart these two probabilities can be.

We are now ready to obtain a communication converse, using the method of \textit{binary inversion}.
\begin{theorem}\label{thm:BI}
    Recall the setting in \S\ref{settings}. Let $h_\alpha(\mu, K)$ be from (\ref{SDPI_bound}), $p_\gamma := 2^{-\gamma k}$, and let $q_\alpha(\gamma) \in (0, p_\gamma]$ such that $
        d_\alpha(p_\gamma || q_\alpha(\gamma)) = kh_\alpha(\mu, K)$.
    Then we have
    \[c \geq C_{\mathrm{BI}}(\gamma) := \sup_{\alpha > 1} \left[ 1 + \frac{1}{k} \log_2 q_\alpha(\gamma) \right]_+ \, . \]
\end{theorem}

\begin{remark}
    There exists a unique $q_\alpha(\gamma)$ that satisfies the hypothesis of the theorem because the map $(0, p_\gamma] \ni q \mapsto d_\alpha(p_\gamma\| q)$ decreases continuously from $+\infty$ to zero. The quantity \(q_\alpha(\gamma)\) is the smallest product-law probability compatible with the constraints \(p\geq p_\gamma\) and \(d_\alpha(p\|q)\leq kh_\alpha\). 
\end{remark}

\begin{proof}[Proof of Theorem \ref{thm:BI}]
    Note that the proof is complete once we observe that $q_\alpha(\gamma) \leq q$, since that would imply
$q_\alpha(\gamma) \leq q \leq 2^{-k(1-c)}$,
and the conclusion of the theorem holds once we take the logarithms of both sides. And note that this is immediate if $q \geq p_\gamma$ since $q_\alpha(\gamma) \leq p_\gamma$. On the other hand, if $q \leq p_\gamma$, then $p \geq p_\gamma \geq q$, and the monotonicity (non-decreasing) of $d_\alpha$ in its first argument gives $d_\alpha(p_\gamma|| q) \leq d_\alpha(p||q).$ And by Lemma \ref{BIE_lemma} and (\ref{SDPI-HC-tensorization}), we have
\begin{align}\label{ultimateBound}
    d_\alpha(p_\gamma || q) \leq d_\alpha(p||q) &\leq k \eta_\alpha(\mu^{\otimes n}, K^{\otimes n})\\ &\leq  kh_\alpha(\mu, K) \,.
\end{align}
Hence, the fact that $d_\alpha$ is decreasing in the second argument on $(0, p_\gamma]$ would guarantee that $q_\alpha(\gamma) \leq q$, so we get the desired result.
\end{proof}

We next obtain an explicit bound directly from Lemma \ref{BIE_lemma} by dropping the complementary term in the binary R\'enyi divergence.

\begin{corollary}\label{thm:caad}
In the setting of \S\ref{settings}, we have
\begin{equation}\label{cLB1}
c\ge
C_{\mathrm{SDPI}}(\gamma) :=\sup_{\alpha>1}
\left[
1-\eta_\alpha(\mu^{\otimes n},K^{\otimes n})
-\alpha'\gamma
\right]_+.
\end{equation}
\end{corollary}

\begin{proof}Let \(p,q\) be as in Lemma \ref{BIE_lemma}. Dropping the complementary term in the binary Rényi divergence gives
\[
\begin{aligned}
k\eta_\alpha(\mu^{\otimes n},K^{\otimes n})
&\ge d_\alpha(p\Vert q)\\
&\ge \alpha'\log_2p-\log_2q\\
&\ge -\alpha'\gamma k+k(1-c).
\end{aligned}
\]
Rearranging, using \(c\ge0\), and taking the supremum over \(\alpha>1\) proves the result. 
\end{proof}

\begin{remark}[Comparison with the hypercontractivity bound]
    The lower bound induced by hypercontractivity \cite[Theorem 3.1.1]{Zhao2022} says that
    \begin{equation}\label{cLB2}
        c \geq C_{\mathrm{HC}}(\gamma) := \sup_{\alpha > 1} \left[ 1 - h_\alpha(\mu, K) - \alpha' \gamma \right]_+ \,.
    \end{equation}
    By (\ref{SDPI-HC-tensorization}), we can witness that the lower bound that we obtained in Corollary \ref{thm:caad} is never weaker than that of (\ref{cLB2}).
\end{remark}

\begin{remark}[The case $\alpha = \infty$]
    By \cite[Corollary 1]{VEG_contraction}, if we let $\alpha \to \infty$ (thus $\alpha' \to 1$) in (\ref{cLB1}), we get
    \begin{equation}\label{cLB1-infty} c \geq 1 - \eta_\infty(\mu^{\otimes n}, K^{\otimes n}) - \gamma \, . \end{equation}
    By contrast, if we let $\alpha \to \infty$ for the hypercontractivity estimate (\ref{cLB2}), note that unless $X$ and $Y$ are independent, we must also let $\beta \to \infty$ by \cite[Theorem 3.1]{KA_2016} (thus $\beta' \to 1$), so it becomes $c \geq -\gamma$, which is  a trivial bound. 
    In the case of (\ref{cLB1-infty}), it is possible to obtain a positive lower bound if $\eta_\infty\left( \mu^{\otimes n}, K^{\otimes n}\right) < 1 - \gamma$. One condition that guarantees that is if $K^{\otimes n}$ is $(\gamma+\epsilon)$-ultra-mixing for some $\epsilon \in (0, 1-\gamma]$, i.e. $\frac{K^{\otimes n}(y|x')}{K^{\otimes n}(y|x)} \geq \gamma + \epsilon$ for all $x, x' \in \mathsf{X}^n$ and $y \in \supp(K^{\otimes n}(\cdot|x))$; this is due to \cite[Theorem 7]{VEG_contraction},  which states that for any Markov chain $M$ and $\delta \in (0, 1)$, $\eta_\infty(M) := \sup_{\mu \in \mathcal{P}(\mathsf{X})} \eta_\infty(\mu, M) \leq 1 - \delta$ if and only if $M$ is $\delta$-ultra-mixing, which is also equivalent to the fact that $M$ satisfies $\ln(1/\delta)$-local differential privacy \cite[Theorem 13]{VEG_contraction}. 
\end{remark}

\subsection{Examples}
Here we provide two examples of the communication converse obtained in Theorem \ref{thm:BI}.
We note that we may approximate the value of $q_\alpha(\gamma)$ from the theorem using bisection: Set $x_\alpha = -\log q_\alpha(\gamma)$. For each $\alpha > 1$, we can approximate the solution
    $d_\alpha(p_\gamma || 2^{-x_\alpha}) = kh_\alpha$ by bisection on the guaranteed bracket $[k\gamma, kh_\alpha + \frac{\alpha}{\alpha -1} k\gamma]$. 
\begin{example}[Doubly symmetric binary symmetric channel]\label{ex:DSBS}
    Consider the doubly symmetric binary symmetric channel (DSBS) of parameter $\lambda \in (0, 1/2)$; we denote this channel $K_\lambda$. Recall that we can compute $\beta_K(\alpha)$ (and $h_\alpha$) from (\ref{DSBS}). 

    If $\alpha = 2$, we can explicitly compute $q_2(\gamma)$; we have $q_2(\gamma) = 2p_\gamma^2/\Delta(\gamma)$, where $\Delta(\gamma) : = 2^{kh_2} + 2p_\gamma - 1 + \sqrt{(2^{kh_2} + 2p_\gamma - 1)^2 - 4 \cdot 2^{kh_2}p_\gamma^2}\,. $    If $\lambda = 0.1$, $\gamma = 0.01$ and $k = 20$, our example gives
    $q_2 \approx 1.516 \times 10^{-5}$, and $1 + \frac{\log_2(q_2(\gamma))}{20} \approx 0.1995$.
    Numerical optimization gives $$C_{\mathrm{BI}}(0.01) \approx 0.2640 > C_{\mathrm{HC}}(\gamma) \approx 0.2604,$$ with the optimum for R\'enyi attained when $\alpha \approx 1.1865$.
    \end{example}

    \begin{example}[Rare-event Z-Channel]\label{ex:Z}
    Given $t \in (0, 1)$, let $X \sim \bern(t)$, $K_t(1|0) = 0$, and $K_t(1|1) = t$. Consequently, we have $Y \sim \bern(t^2)$, and the nonzero joint probabilities are $P_{XY}(0, 0) = 1-t$, $P_{XY}(1, 0) = t(1-t)$, and $P_{XY}(1, 1) = t^2$.

    Fig. \ref{fig:z-channel} shows the numerical computation of the communication converse between the binary-inversion (Theorem \ref{thm:BI}) and hypercontractivity (\ref{cLB2}) for the case when $t \approx 0.048$, $n = 56$ and $k = 4$. Since we are working on a binary system, the numerical value for the hypercontractivity bound was obtained by approximating (via bisection) $\beta_{K_t}(\alpha) =  \inf\bigl\{\beta \in [1, \alpha): \frac{1}{\alpha'} d_\alpha(tr || t^2) \leq \frac{1}{\beta'}d_\beta(r||t)  \text{ for every } r \in [0, 1] \bigr\}$.
    \end{example}

\begin{figure}[tbp]
    \centering
    \includegraphics[width=0.9\linewidth]
        {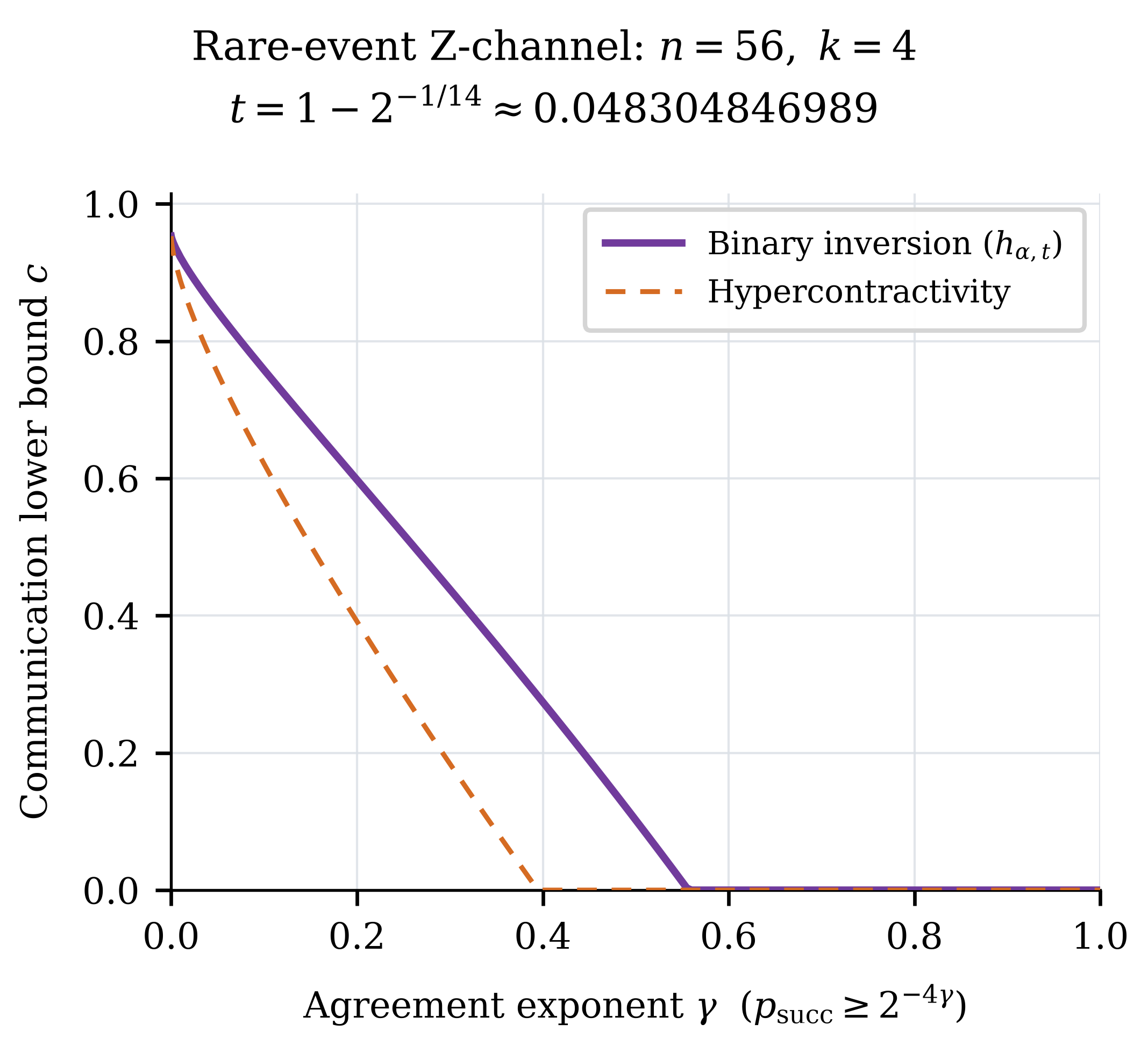}
    \caption{Comparison between the bound given by the binary-inversion SDPI (Thm. \ref{thm:BI}) and  hypercontractivity (\ref{cLB2}) for the Z-channel defined in Example~\ref{ex:Z} with $t \approx 0.0483$, $n = 56$ and $k=4$.}
    \label{fig:z-channel}
\end{figure}
\newpage

\section*{AI Statement}
The authors used OpenAI’s Codex for language editing and expository suggestions throughout the manuscript, mathematical consistency checks, literature and reference checks, and numerical verification for Examples \ref{ex:DSBS} and \ref{ex:Z} (including the production of Figure \ref{fig:z-channel}). AI was also used to generate this statement. The authors reviewed and edited the AI-assisted material and take full responsibility for the content of the paper.
\printbibliography
\end{document}